\RequirePackage{xifthen}
\newcommand{\inShortened}[1]{{\ifthenelse{\isnamedefined{shorten}}{#1}{}}}
\newcommand{\inFull}[1]{{\ifthenelse{\isnamedefined{shorten}}{}{#1}}}

\documentclass{article}
\usepackage[paper=letterpaper,margin=1in]{geometry}
\usepackage{amsfonts}       
\usepackage{nicefrac}       
\usepackage{forloop}
\usepackage{amsthm}
\usepackage{amsmath}
\usepackage{thmtools}
\usepackage[ruled,vlined,linesnumbered]{algorithm2e}
\usepackage[numbers]{natbib}
\usepackage{hyperref}       
\usepackage{url}            

\hypersetup{
  colorlinks   = true, 
  urlcolor     = blue, 
  linkcolor    = blue, 
  citecolor   = red 
}

\newtheorem{theorem}{Theorem}[section]
\newtheorem{lemma}[theorem]{Lemma}

\newtheorem{corollary}[theorem]{Corollary}

\newtheorem{observation}[theorem]{Observation}

\newcommand{\defcal}[1]{\expandafter\newcommand\csname c#1\endcsname{{\mathcal{#1}}}}
\newcommand{\defbb}[1]{\expandafter\newcommand\csname b#1\endcsname{{\mathbb{#1}}}}
\newcommand{\defvec}[1]{\expandafter\newcommand\csname v#1\endcsname{{\mathbf{#1}}}}
\newcounter{calBbCounter}
\forLoop{1}{26}{calBbCounter}{
    \edef\capital{\Alph{calBbCounter}}
    \expandafter\defcal\capital
		\expandafter\defbb\capital
		\edef\lowerLetter{\alph{calBbCounter}}
		\expandafter\defvec\lowerLetter
}

\newcommand{\real}{\mathbb{R}_+}

\newcommand{\XM}{\mathcal{X}}
\newcommand{\eps}{\varepsilon}
\newcommand{\vzero}{{\mathbf{0}}}
\newcommand{\iter}[2]{{{#1}^{(#2)}}}
\newcommand{\iters}[3]{{{#1}^{(#2)}_{#3}}}
\newcommand{\OPT}{{\mathsf{opt}}}
\newcommand{\AlgCA}{\textsc{Coordinate-Ascent}\xspace}

\newcommand{\AlgECA}{\textsc{Coordinate-Ascent+}\xspace}
\newcommand{\AlgFECA}{\textsc{Coordinate-Ascent++}\xspace}
\newcommand{\basic}{{\ve}}
\newcommand{\outputCa}{{\vx_{\text{\textnormal{CA}}}}}
\newcommand{\outputCaP}{{\vx_{\text{\textnormal{CA+}}}}}
\newcommand{\outputCaPP}{{\vx_{\text{\textnormal{CA++}}}}}

\title{Continuous Submodular Maximization: Beyond DR-Submodularity}

\author{%
  Moran Feldman\thanks{Department of Computer Science, University of Haifa. E-mail: \texttt{moranfe@cs.haifa.ac.il}}
  \and
  Amin Karbasi\thanks{School of Engineering and Applied Science, Yale University. E-mail: \texttt{amin.karbasi@yale.edu}}
}

\begin{document}

\maketitle

\begin{abstract}
  In this paper, we propose the first continuous optimization algorithms that achieve a constant factor approximation guarantee for the problem of monotone continuous submodular maximization subject to a linear constraint. We first prove that a simple variant of the vanilla coordinate ascent, called \AlgECA, achieves a $(\frac{e-1}{2e-1}-\eps)$-approximation guarantee while performing $O(n/\eps)$ iterations, where the computational complexity of each iteration is roughly $O(n/\sqrt{\eps}+n\log n)$ (here, $n$ denotes the dimension of the optimization problem). We then propose \AlgFECA, that achieves the tight $(1-1/e-\eps)$-approximation guarantee while performing the same number of iterations, but at a higher computational complexity of roughly $O(n^3/\eps^{2.5} + n^3 \log n / \eps^2)$ per iteration. However,   the computation of each round of \AlgFECA can be easily parallelized so that the computational cost per machine scales as $O(n/\sqrt{\eps}+n\log n)$.
\end{abstract}

\section{Introduction}
Submodularity is a fundamental concept in combinatorial optimization,  usually associated with discrete set functions \citep{fujishige91old}.  As submodular functions formalize the intuitive notion of diminishing returns, and thus provide a useful structure,  they appear in a wide range of  modern machine learning applications including various forms of data summarization \citep{lin2012learning, mirzasoleiman2013distributed}, influence maximization \citep{kempe03},  sparse and deep representations \citep{balkanski2016learning,oh2017deep}, fairness \cite{celis2016fair, kazemi2018scalable}, experimental design \citep{harshaw2019submodular}, neural network interpretability \cite{elenbergDFK17}, human-brain mapping \citep{salehi2017submodular}, adversarial robustness \citep{lei2018discrete}, crowd teaching \citep{singla2014near},  to name a few.
Moreover, submodularity  ensures the  tractability of the underlying combinatorial optimization problems as minimization of submodular functions can be done exactly and (constrained) maximization of submodular functions can be done approximately. 

To capture an even larger set of applications, while providing rigorous guarantees,  the discrete notion of submodularity has been generalized in various directions, including adaptive and interactive submodualarity for sequential decision making problems \citep{golovin11, guillory10interactive}, weak submodularity for general set functions with a bounded submodularity distance \citep{das2011submodular} and sequence submodularity for  time series analysis \citep{tschiatschek2017selecting, mitrovic2019adaptive}, among other variants.

Very recently, a surge of new applications   in machine learning and statistics motivated researchers to study continuous submodular functions \citep{bach2015submodular, wolsey1982analysis}, a large class of non-convex/non-concave functions, which may be optimized efficiently.  In particular, it has been shown that  continuous submodular minimization can be done exactly \cite{bach2015submodular}. In contrast, for   continuous submodular maximization, it is usually assumed that the continuous function is not only submodular, but also has the extra condition of diminishing returns. Such functions are usually called continuous DR-submodular \citep{bian2016guaranteed}.  We should highlight that even though in the discrete domain, submodularity and diminishing returns are equivalent; in the continuous domain, the diminishing returns condition  implies continuous submodularity, but not vice versa. 

In this paper, we propose the first algorithms that achieve constant factor approximation guarantees for the maximization of a monotone continuous submodular function subject to a linear constraint. More specifically, \textbf{our contributions} can be summarized as follows:
\begin{itemize}
	\item We develop a variant of the coordinate ascent algorithm, called \AlgECA, that achieves a $(\frac{e-1}{2e-1}-\eps)$-approximation guarantee while performing $O(n/\epsilon)$ iterations, where the computational complexity of each iteration is  $O(n\sqrt{B/\eps}+n\log n)$. Here, $n$ and $B$ denote the dimension of the optimization problem and the $\ell_1$ radius of the  constraint set, respectively. 
	
	\item We  then develop \AlgFECA, that achieves the tight $(1-1/e-\eps)$ approximation guarantee  while performing $O(n/\epsilon)$ iterations, where the computational complexity of each iteration is $O(n^3\sqrt{B}/\eps^{2.5} + n^3 \log n / \eps^2)$. Moreover,   \AlgFECA can be easily parallelized so that the computational complexity  per machine in each round scales as $O(n\sqrt{B/\epsilon}+n\log n)$.
\end{itemize}
 Notably, to establish these results,  we do not assume that the continuous submodular function satisfies the diminishing returns condition.

\subsection{Related Work}
 Continuous submodular functions naturally arise in many machine learning applications such as Adwords for e-commerce and advertising \citep{mehta2007adwords, devanur2012online},  influence and revenue maximization \citep{bian2016guaranteed}, robust budget allocation \citep{staib2017robust},  multi-resolution data summarization \citep{bian2016guaranteed}, learning assignments \citep{golovin2014online},  experimental design \citep{chen2018online}, and MAP inference for determinantal point processes \citep{gillenwater2012near, hassani2019stochastic}. Continuous submodular functions have also been studied in statistics as negative
log-densities of probability distributions. These distributions are referred to as multivariate
totally positive of order 2 (MTP2) \citep{fallat2017total} and classical examples are the multivariate logistic, Gamma and $F$
distributions, as well as characteristic roots of random Wishart matrices \citep{karlin1980classes}.

The focus of the current work is to study continuous submodular maximization. Almost all the existing works in this area consider a  proper subclass of continuous submodular functions, called continuous DR-submodular, which satisfy diminishing returns conditions. In particular, when first order information (i.e., exact or stochastic gradients) is available  \citep{hassanigradient2017} showed that (stochastic) gradient ascent achieves $1/2$-approximation guarantee for monotone continuous DR-submodular functions subject to a general convex body constraint. Interestingly, one can achieve the tight approximation guarantee of $1-1/e$ by using conditional gradient methods \citep{bian2016guaranteed} or its efficient stochastic variants \citep{mokhtari2018conditional, NIPS2019hassani, zhang2019sample}.  A simple variant of the conditional gradient methods can also be applied to  non-monotone DR-submodular functions, which results in a $1/e$-approximation guarantee \citep{NIPS2017_6652, mokhtari2018stochastic, hassani2019stochastic}.   The only work, we are aware of, that goes beyond the above line of work, and considers also non-DR continuous submodular functions is a recent work by \citet{niazadeh2018optimal}, which developed a polynomial time algorithm with a tight $1/2$-approximation guarantee for the problem of continuous submodular maximization subject to a box constraint.

Discrete and continuous submodular maximization problems are inherently related to one another through the multilinear extension \citep{calinescu11maximizing}. Indeed, maximization of the multilinear extension (along with a subsequent rounding) has led to the best theoretical results in many settings, including submodular maximization subject to various complex constraints \citep{feldman2011unified,chekuri2014submodular,buchbinder2019constrained}, online and bandit submodular maximization \citep{zhang2019online, chen2018projection}, decentralized solution \citep{mokhtari2018decentralized, xie2019decentralized}, and algorithms with low adaptivity complexity \cite{chekuri2019parallelizing,balkanski2019optimal, chen2019unconstrained, ene2019submodular}.  
\section{Preliminaries and Problem Formulation}
We first recall a few standard definitions regarding submodular functions. Even though 
submodularity is mostly considered in the discrete 
domain, the notion can
be naturally extended to arbitrary lattices 
\citep{fujishige91old}. To this end, let us consider a 
subset of 
$\real^d$ of the form $\XM = \prod_{i=1}^d 
\XM_i$ where each $\XM_i$ is a compact subset of 
$\real$. A function $F\colon \XM\rightarrow \real$ is  
\textit{submodular} \citep{wolsey82} if 
for all $(\vx,\vy)\in \XM \times \XM$, we have
\begin{equation*}
F(\vx)+ F(\vy) \geq F(\vx\vee \vy) + F(\vx\wedge \vy) \enspace,
\end{equation*}
where $\vx\vee \vy \doteq\max(\vx,\vy)$ (component-wise) and $\vx\wedge \vy \doteq \min (\vx,\vy)$ (component-wise). A submodular function is monotone if for any $\vx,\vy\in\XM$ such that $\vx\leq \vy$, we have $F(\vx)\leq F(\vy)$ (here, by $\vx \leq \vy$ we mean that every element of $\vx$ is less than that of $\vy$). The above definition includes the discrete notion of submodularity over a set by restricting each $\XM_i$ to $\{0,1\}$. In this paper, we mainly consider \textit{continuous submodular} functions, where each $\XM_i$ is a closed interval in $\real$.  When $F$ is twice differentiable, a continuous function is submodular if and only if all cross-second-derivatives are non-positive \citep{bach2015submodular}, i.e., 
\begin{equation*}
\forall i\neq j, \forall \vx\in \XM, ~~ \frac{\partial^2 F(\vx)}{\partial x_i \partial x_j} \leq 0 \enspace.
\end{equation*}
Thus, continuous 
submodular functions can be convex (e.g., $F(\vx) = \sum_{i,j} \phi_{i,j}(x_i-x_j)$ for $\phi_{i,j}$ convex), concave (e.g., $F(\vx)=g(\sum_{i=1}^n\lambda_i x_i)$ for $g$ concave and $\lambda_i$'s non-negative), and neither (e.g.,  quadratic program $F(\vx)=\vx^TQ \vx$ where all off-diagonal elements of $Q$
are non-positive). 

A proper subclass of continuous submodular functions are called \textit{DR-submodular} \citep{bian2016guaranteed, soma2015generalization} if for all $\vx,\vy\in\XM$ such that $\vx\leq \vy$, standard basis vector $\basic_i\in\real^n$ and a non-negative number $z\in\real$ such that $z\basic_i+\vx\in\XM$ and $z \basic_i+\vy\in\XM$, it holds that
$F(z \basic_i+\vx)-F(\vx)\geq F(z\basic_i+\vy)-F(\vy). $
One can easily verify that for a differentiable DR-submodular function the gradient is an antitone mapping, i.e., for all $\vx,\vy \in \XM$ such that $\vx\leq \vy$ we have  $\nabla F(\vx) \geq \nabla F(\vy)$ \citep{bian2016guaranteed}. An important example of a DR-submodular function is the multilinear extension \citep{calinescu2011maximizing}. 

In this paper, we consider the following fundamental optimization problem
\begin{eqnarray}\label{eq:problem}
\max_{\vx\in\XM} F(\vx) \text{ subject to } \lVert \vx \rVert_1\leq B \enspace,
\end{eqnarray}
where $F$ is a non-negative monotone continuous submodular function. Without loss of generality, we assume that each closed interval $\XM_i$ is of the form $[0,u_i]$ since otherwise for $\XM_i = [a_i, a_i+u_i]$ we can always define a corresponding continuous submodular function  $G(\vx) = F(\vx+\va)$, where $\va = [a_1, \dots, a_n]$. Similarly, we assume w.l.o.g., that $u_i\leq B$ for every coordinate $i$. We also assume that $F$ is $L$-smooth, meaning that $\|\nabla F(\vx) - \nabla F(\vy)\|_2 \leq L \|\vx-\vy\|_2$ for some $L\geq 0$ and for all $\vx,\vy\in\XM$. Finally, note that replacing a linear constraint of the form $\sum_{i = 1}^n w_i x_i \leq B$ (where $w_i>0$) with $\lVert \vx \rVert_1\leq B$ does not change the nature of the problem. In this case, we can simply define a corresponding function  $G(\vx)=F\left(\sum_{i = 1}^n x_i \basic_i / w_i\right)$ and solve Problem~\eqref{eq:problem}. This change of course changes $L$ by a factor of $W = \min_{1 \leq i \leq n} w_i$. 
Prior to our work, no constant approximation guarantee  was known for Problem~\eqref{eq:problem}. 

\section{Plain Coordinate Ascent} \label{sec:plain}

In this section we present our plain coordinate ascent algorithm and analyze its guarantee. Our algorithm uses as a black box an algorithm for a one dimensional optimization problem whose properties are summarized by the following proposition. We include the proof of this proposition in Appendix~\ref{app:one_coordinate_ratio}.

\begin{restatable}{proposition}{propOneCoordinateRatio} \label{prop:one_coordinate_ratio}
Given a point $\vx \in [\vzero, \vu]$, a coordinate $i \in [n]$, bounds $0 < a \leq b \leq u_i - x_i$ and a positive parameter $\eps \in (0, 1)$ there is a polynomial time algorithm that runs in $O(\sqrt{B / \eps} + \log (\eps / a))$ time and returns a value $y \in [a, b]$ maximizing the ratio $F(\vx + y\basic_i) / y$ up to an additive error of $\eps L$.
\end{restatable}

Using the algorithm whose existence is guaranteed by the last proposition, we can now formally state our coordinate ascent algorithm as Algorithm~\ref{alg:greedy}. This algorithm gets a quality control parameter $\eps \in (0, \nicefrac{1}{4})$.

\begin{algorithm}
\DontPrintSemicolon
\caption{\AlgCA$(\eps)$} \label{alg:greedy}
Let $\vx \gets \vzero$ and $\delta \gets \eps B / n$.\\
\While{$\|\vx\|_1 \leq B$}
{
	Let $C \subseteq [n]$ be the set of coordinates $i \in [n]$ for which $x_i < u_i$ (i.e., these coordinates can be increased in $\vx$ to some positive extent without violating feasibility).\\
	\For{every $i \in C$}
	{
		Let $d'_i$ be the maximum amount by which $x_i$ can be increased without violating feasibility. Formally, $d'_i = \min\{u_i - x_i, B - \|\vx\|_1\}$.\\
		Use the algorithm suggested by Proposition~\ref{prop:one_coordinate_ratio} to find a value $d_i \in [\min\{d'_i, \delta\}, d'_i]$ maximizing $\frac{F(\vx + d_i\basic_i) - F(\vx)}{d_i}$ up to an additive error of $\eps L$.
	}
	Let $j$ be the coordinate of $C$ maximizing $\frac{F(\vx + d_j\basic_j) - F(\vx)}{d_j}$, and update $\vx \gets \vx + d_j\basic_j$.
}
\Return{$\vx$}.
\end{algorithm}

We begin the analysis of Algorithm~\ref{alg:greedy} with the following observation that bounds its time complexity.\inShortened{ The proof of this observation and some other proofs appearing later in this paper have been moved to Appendix~\ref{app:missing} due to space constraints. In a nutshell, the observation holds since each iteration of the main loop of Algorithm~\ref{alg:greedy} producees progress in at least one of three ways: making $x_j = u_j$ for some coordinate $j$, making $\|\vx\|_1$ as large as $B$, or increasing $\|\vx\|_1$ by at least $\delta$.}
\begin{restatable}{observation}{obsComplexityGreedy} \label{obs:complexity_greedy}
The main loop of Algorithm~\ref{alg:greedy} makes at most $O(n/\eps)$ iterations, and each iteration runs in $O(n\sqrt{B/\eps} + n \log n)$ time. Thus, the entire algorithm runs in $O(n^2\sqrt{B}/\eps^{1.5} + n^2 \log n / \eps)$ time.
\end{restatable}
\newcommand{\proofComplexityGreedy}{
\begin{proof}
We note that the way in which the algorithm assigns a value to $d_j$ implies that in any iteration of the main loop of Algorithm~\ref{alg:greedy} one of the following must happen.
\begin{enumerate}
	\item One option is that $d_j = u_j - x_j$. When this happens, the value of $x_j$ becomes equal to $u_j$, and thus, this is the last iteration in which the coordinate $j$ belongs to the set $C$. \label{item:element_equality}
	\item Another option is that $d_j = B - \|\vx\|_1$. In this case, $\|\vx\|_1$ becomes equal to $B$ following the iteration, and thus, the algorithm terminates following this iteration. \label{item:constraint_tightness}
	\item If neither of the previous options happens, then the value $\|\vx\|_1$ increases by at least $\delta$ following the iteration. \label{item:delta_increase}
\end{enumerate}
There can be at most $n$ iterations in which Option~\ref{item:element_equality} happens since there are only $n$ coordinates, at most a single iteration in which Option~\ref{item:constraint_tightness} happens and at most $B/\delta = n / \eps$ iterations in which Option~\ref{item:delta_increase} happens (since the value of $\|\vx\|_1$ cannot exceed $B$). Thus, the total number of iterations is at most
\[
	n + 1 + \frac{n}{\eps}
	=
	O(\eps^{-1} n)
	\enspace.
\]

We now note that every single iteration of the main loop of Algorithm~\ref{alg:greedy} requires $O(n)$ time plus the time required for up to $n$ executions of the algorithm whose existence is guaranteed by Proposition~\ref{prop:one_coordinate_ratio}. Furthermore, we can assume that each execution of the last algorithm gets $a = \delta$ because we always look for $d_i$ either inside a range containing a single value or a range whose lower bound is $\delta$. Thus, the time required for each such execution is upper bounded by
\[
	O\left(\sqrt{\frac{B}{\eps}} + \log\left(\frac{\eps}{\delta}\right)\right)
	=
	O\left(\sqrt{\frac{B}{\eps}} + \log\left(\frac{n}{B}\right)\right)
	=
	O(B/\eps^{0.5} + \log n)
	\enspace,
\]
and the space required for the entire iteration of the main loop of Algorithm~\ref{alg:greedy} is at most
\[
	n \cdot O(B/\eps^{0.5} + \log n) + O(n)
	=
	O(nB/\eps^{0.5} + n \log n)
	\enspace.
	\tag*{\qedhere}
\]
\end{proof}}
\inFull{\proofComplexityGreedy}

Fix now some feasible solution $y \in [\vzero, \vu]$. Intuitively, we say that an iteration of the main loop of Algorithm~\ref{alg:greedy} is \emph{good} (with respect to $\vy$) if, at the beginning of the iteration, the algorithm still has the option to increase each coordinate of $\vx$ to be equal to the corresponding coordinate of $\vy$, and this does not violate the constraint. Formally, an iteration is good  if the inequality $y_i - x_i \leq d'_i$ was true in this iteration for every coordinate $i \in C$ (before the vector $\vx$ was updated at the end of the iteration). Let $\ell$ denote the number of good iterations of the main loop of Algorithm~\ref{alg:greedy}, and let us denote by $\iter{\vx}{h}$ the value of $\vx$ after $h$ iterations for every $0 \leq h \leq \ell$. Using this notation, we can now state and prove the following lemma, which provides a lower bound on the value of $\vx$ after any  number of (good) iterations of Algorithm~\ref{alg:greedy}.

\begin{lemma} \label{lem:conditioned_guarantee}
For every vector $\vy \in [\vzero, \vu]$ and integer $0 \leq h \leq \ell$, $F(\iter{\vx}{h}) \geq (1 - e^{-\|\iter{\vx}{h}\|_1/(\|y\|_1 + \eps B)}) \cdot F(\vy) - \|\iter{\vx}{h}\|_1 \cdot \eps L$.
\end{lemma}
\begin{proof}
We prove the lemma by induction on $h$. For $h = 0$, $\|\iter{\vx}{h}\|_1 = 0$, and the lemma follows from the non-negativity of $F$. Thus, it remains to prove the lemma for some $h > 0$ given that it holds for $h - 1$. From this point on we restrict our attention to iteration number $h$ of Algorithm~\ref{alg:greedy}, and thus, when we refer to variables such as $C$ and $d_i$, these variables should be understood as taking the values they are assigned in this iteration. Given this assumption, for every $i \in C$, let us now define a value $o_i$ that is closest to $y_i - \iters{x}{h - 1}{i}$ among all the values in the range to which $d_i$ can belong. Formally,
\[
	o_i
	=
	\min\{\max\{y_i - \iters{x}{h - 1}{i}, \min\{\delta, d'_i\}\}, d'_i\}
	=
	\max\{y_i - \iters{x}{h - 1}{i}, \min\{\delta, d'_i\}\}
	\quad
	\forall\; i \in C
	\enspace,
\]
where the equality holds since the fact that the iteration we consider is a good iteration implies $y_i - \iters{x}{h - 1}{i} \leq d'_i$. Since $o_i$ is a valid choice for $d_i$, we get by the definition of $d_i$ that
\[
	\frac{F(\iter{\vx}{h - 1} + d_i\basic_i) - F(\iter{\vx}{h - 1})}{d_i}
	\geq
	\frac{F(\iter{\vx}{h - 1} + o_i\basic_i) - F(\iter{\vx}{h - 1})}{o_i} - \eps L
	\enspace.
\]
Using the definition of $j$ and the submodularity of $F$, the last inequality implies
{\allowdisplaybreaks
\begin{align} \label{eq:j_average}
	&
	\frac{F(\iter{\vx}{h - 1} + d_j\basic_j) - F(\iter{\vx}{h - 1})}{d_j}
	\geq
	\frac{\sum_{i \in C} o_i \cdot \frac{F(\iter{\vx}{h - 1} + d_i\basic_i) - F(\iter{\vx}{h - 1})}{d_i}}{\sum_{i \in C} o_i} \\ \nonumber
	\geq{} &
	\inFull{\frac{\sum_{i \in C} o_i \cdot \left(\frac{F(\iter{\vx}{h - 1} + o_i\basic_i) - F(\iter{\vx}{h - 1})}{o_i} - \eps L\right)}{\sum_{i \in C} o_i}
	=}
	\frac{\sum_{i \in C} [F(\iter{\vx}{h - 1} + o_i\basic_i) - F(\iter{\vx}{h - 1})]}{\sum_{i \in C} o_i} - \eps L \\ \nonumber
	\geq{} &
	\frac{F(\iter{\vx}{h - 1} + \sum_{i \in C} o_i\basic_i) - F(\iter{\vx}{h - 1})]}{\sum_{i \in C} o_i} - \eps L
	\enspace.
\end{align}%
}%

To understand the rightmost side of the last inequality, we need the following two bounds.
\[
	\sum_{i \in C} o_i
	\leq
	\sum_{i \in C} \max\{y_i, \delta\}
	\leq
	\sum_{i \in C} y_i + n\delta
	\leq
	\|\vy\|_1 + \eps B
	\enspace,
\]
and
\[
	\iter{\vx}{h - 1} + \sum_{i \in C} o_i\basic_i
	\geq
	\iter{\vx}{h - 1} + \sum_{i \in C} (y_i - \iters{x}{h - 1}{i})\basic_i
	\geq
	\vy
	\enspace.
\]

Plugging these bounds into Inequality~\eqref{eq:j_average}, and using the monotonicity of $F$, we get
\[
	\frac{F(\iter{\vx}{h - 1} + d_j\basic_j) - F(\iter{\vx}{h - 1})}{d_j}
	\geq
	\frac{F(\vy) - F(\iter{\vx}{h - 1})]}{\|\vy\|_1 + \eps B} - \eps L
	\enspace.
\]
Since $\iter{\vx}{h} = \iter{\vx}{h - 1} + d_j\basic_j$, the last inequality now yields the following lower bound on $F(\iter{\vx}{h})$.
\begin{align*}
	F(\iter{\vx}{h})
	={} &
	F(\iter{\vx}{h - 1}) + [F(\iter{\vx}{h}) - F(\iter{\vx}{h - 1})]\\
	\geq{} &
	F(\iter{\vx}{h - 1}) + \frac{d_j}{\|\vy\|_1 + \eps B} \cdot [F(\vy) - F(\iter{\vx}{h - 1})] - \eps L d_j\\
	\geq{} &
	\left(1 - \frac{d_j}{\|\vy\|_1 + \eps B}\right) \cdot F(\iter{\vx}{h - 1}) + \frac{d_j}{\|\vy\|_1 + \eps B} \cdot F(\vy) - \eps L d_j
	\enspace.
\end{align*}
Finally, plugging into the last inequality the lower bound on $F(\iter{\vx}{h - 1})$ given by the induction hypothesis, we get
\begin{align*}
	F(\iter{\vx}{h})
	\geq{} &
	\left(1 - \frac{d_j}{\|\vy\|_1 + \eps B}\right) \cdot \left\{(1 - e^{-\|\iter{\vx}{h - 1}\|_1/(\|\vy\|_1 + \eps B)}) \cdot F(\vy) - \|\iter{\vx}{h - 1}\|_1 \cdot \eps L\right\} \\& + \frac{d_j}{\|\vy\|_1 + \eps B} \cdot F(\vy) - \eps L d_j\\
	\geq{} &
	\left(1 - \left(1 - \frac{d_j}{\|\vy\|_1 + \eps B}\right) \cdot e^{-\|\iter{\vx}{h - 1}\|_1/(\|\vy\|_1 + \eps B)}\right) \cdot F(\vy) - (\|\iter{x}{h - 1}\|_1 + d_j) \cdot \eps L\\
	\geq{} &
	\left(1 - e^{-(\|\iter{x}{h - 1}\|_1 + d_j)/(\|\vy\|_1 + \eps B)}\right) \cdot F(\vy) - (\|\iter{x}{h - 1}\|_1 + d_j) \cdot \eps L\\
	={} &
	\left(1 - e^{-\|\iter{x}{h}\|_1/(\|\vy\|_1 + \eps B)} \right) \cdot F(\vy) - \|\iter{x}{h}\|_1 \cdot \eps L
	\enspace.
	\tag*{\qedhere}
\end{align*}
\end{proof}

Our next objective is to get an approximation guarantee for Algorithm~\ref{alg:greedy} based on the last lemma. Such a guarantee appears below as Corollary~\ref{cor:plain_greedy}. However, to prove it we also need the following observation, which shows that lower bounding the value of $F(x)$ at some point during the execution of Algorithm~\ref{alg:greedy} implies the same bound also for the value of the final solution of the algorithm.

\begin{observation} \label{obs:increase}
The value of $F(\vx)$ only increases during the execution of Algorithm~\ref{alg:greedy}.
\end{observation}
\begin{proof}
The observation follows from the monotonicity of $F$ since $d_j$ is always non-negative.
\end{proof}

Let $\OPT$ be some optimal solution vector.
\begin{restatable}{corollary}{corPlainGreedy} \label{cor:plain_greedy}
Let $\outputCa$ be the vector outputted by Algorithm~\ref{alg:greedy}, then
$
	F(\outputCa)
	\geq
	(1 - 1/e - B^{-1} \cdot \max\nolimits_{i \in [n]} u_i - \eps) \cdot F(\OPT) - \eps B L
$.
\end{restatable}
\newcommand{\proofPlainGreedy}{
\begin{proof}
By Observation~\ref{obs:increase}, it suffices to argue that
\[
	F(\iter{\vx}{\ell})
	\geq
	(1 - 1/e - B^{-1} \cdot \max\nolimits_{i \in [n]} u_i - \eps) \cdot F(\OPT) - \eps B L
	\enspace.
\]
Thus, in the rest of the proof we prove this inequality.

Plugging $\vy = \OPT$ into Lemma~\ref{lem:conditioned_guarantee}, we get
\begin{equation} \label{eq:general_ell_prime}
	F(\iter{\vx}{\ell})
	\geq
	(1 - e^{-\|\iter{\vx}{\ell}\|_1/(\|\vy\|_1 + \eps B)}) \cdot F(\OPT) - \|\iter{\vx}{\ell}\|_1 \cdot \eps L
	\geq
	(1 - e^{-(1 - \eps)\|\iter{\vx}{\ell}\|_1/B}) \cdot F(\OPT) - \eps BL
	\enspace,
\end{equation}
where the second inequality holds since $\|\vy\|_1$ and $\|\iter{\vx}{\ell}\|_1$ are both upper bounded by $B$. If iteration number $\ell$ is not the last iteration of Algorithm~\ref{alg:greedy}, then the fact that iteration number $\ell + 1$ was not a good iteration implies the existence of a coordinate $i \in [n]$ such that $y_i - \iters{x}{\ell}{i} > d'_i = B - \|\iter{\vx}{\ell}\|_1$ (the last equality holds since the inequalities $y_i - \iters{x}{\ell}{i} > d'_j$ and $u_i \geq y_i$ exclude the possibility of $d'_i = u_i - \iters{x}{\ell}{i}$). Thus, we get in this case
\[
	\|\iter{\vx}{\ell}\|_1
	>
	B - y_i + \iters{x}{\ell}{i}
	\geq
	B - u_i
	\geq
	B - \max\nolimits_{i \in [n]} u_i
	\enspace.
\]
Moreover, the last inequality holds also in the case in which iteration number $\ell$ is the last iteration of Algorithm~\ref{alg:greedy} because in this case $\|\iter{\vx}{\ell}\|_1 = B$. Plugging this into Inequality~\eqref{eq:general_ell_prime}, we get
\begin{align*}
	F(\iter{\vx}{\ell})
	\geq{} &
	(1 - e^{(1 - \eps)(\max_{i \in [n]} u_i/B - 1)}) \cdot F(\OPT) - \eps BL\\
	\geq{} &
	(1 - e^{\eps + \max_{i \in [n]} u_i/B - 1}) \cdot F(\OPT) - \eps BL
	\geq
	(1 - e^{- 1} - B^{-1} \cdot \max\nolimits_{i \in [n]} u_i - \eps) \cdot F(\OPT) - \eps BL
	\enspace,
\end{align*}
where the last inequality holds since $e^{x - 1} \leq e^{-1} + x$ for $x \in [0, 1.5]$.
\end{proof}%
}
\inFull{\proofPlainGreedy}

The guarantee of Corollary~\ref{cor:plain_greedy} is close to an approximation ratio of $1 - 1/e$ when the upper bound $u_i$ is small compared to $B$ for every $i \in [n]$. In the next two sections we describe enhanced versions of our coordinate ascent algorithm that give an approximation guarantee which is independent of this assumption. We note that, formally, the analyses of these enhanced versions are independent of Corollary~\ref{cor:plain_greedy}. However, the machinery used to prove this corollary is reused in these analyses.
\section{Fast Enhanced Coordinate Ascent}

In this section we describe one simple and fast way to enhance the plain coordinate ascent algorithm from Section~\ref{sec:plain}, leading to the algorithm that we name \AlgECA. Before describing \AlgECA itself, let us give a different formulation for the guarantee of Algorithm~\ref{alg:greedy}.
\begin{lemma} \label{lem:element_out}
There is a coordinate $j \in [n]$ such that the output $\outputCa$ of Algorithm~\ref{alg:greedy} has a value of at least $(1 - 1/e - 2\eps) \cdot F(\OPT - \OPT_j \basic_j) - \eps BL$.
\end{lemma}
\begin{proof}
In this proof we use the notation from Section~\ref{sec:plain}, and consider the last iteration $\ell'$ during this execution in which there is no coordinate $i \in [n]$ such that $\OPT_i - x_i > d'_i$ (where $x_i$ represents here its value at the beginning of the iteration). If $\ell'$ is the last iteration of Algorithm~\ref{alg:greedy}, then all the iterations of Algorithm~\ref{alg:greedy} are good when we choose $\vy = \OPT$. Thus, for this choice of $\vy$ we get $\|\iter{\vx}{\ell'}\|_1 = B$, and by Lemma~\ref{lem:conditioned_guarantee} the value of the output $\outputCa = \iter{\vx}{\ell'}$ of Algorithm~\ref{alg:greedy} is at least
\[
	(1 - e^{-B/(\|\OPT\|_1 + \eps B)}) \cdot F(\OPT) - \eps B L
	\geq
	(1 - e^{\eps - 1}) \cdot F(\OPT) - \eps B L
	\geq
	(1 - e^{-1} - \eps) \cdot F(\OPT) - \eps BL
	\enspace,
\]
where the \inFull{second }inequality holds since $\|\OPT\|_1 \leq B$. This guarantee is stronger than the guarantee of the lemma (because of the monotonicity of $F$), and thus, completes the proof for the current case.

Consider now the case in which iteration $\ell'$ is not the last iteration of Algorithm~\ref{alg:greedy}. In this case we set $j$ to be some coordinate in $[n]$ for which the inequality $\OPT_j - \iters{x}{\ell'}{j} > d'_j$ holds. Choosing $\vy = \OPT - \OPT_j \basic_j$, we get that Algorithm~\ref{alg:greedy} has at least $\ell'$ good iterations. There are now two cases to consider based on the relationship between $\|\vy\|_1$ and $B$. If $\|\vy\|_1 \geq B/2$, then Lemma~\ref{lem:conditioned_guarantee} and Observation~\ref{obs:increase} imply together that the value of $\outputCa$  is at least
\begin{align*}
	F(\outputCa)
	\geq{} &
	F(\iter{\vx}{\ell'})
	\geq
	(1 - e^{-\|\iter{\vx}{\ell'}\|_1/(\|\vy\|_1 + \eps B)}) \cdot F(\OPT - \OPT_j \basic_j) - \|\iter{\vx}{\ell'}\|_1 \cdot \eps L\\
	\geq{}&
	(1 - e^{-(B - \OPT_j)\|_1/(B + \eps B - \OPT_j)}) \cdot F(\OPT - \OPT_j \basic_j) - \eps BL\\
	\geq{} &
	(1 - e^{2\eps - 1}) \cdot F(\OPT - \OPT_j \basic_j) - \eps BL
	\geq
	(1 - e^{- 1} - 2\eps) \cdot F(\OPT - \OPT_j \basic_j) - \eps BL
	\enspace,
\end{align*}
where the \inFull{third}\inShortened{second} inequality holds since $\|\iter{\vx}{\ell'}| \leq B$, but $\OPT_j - \iters{x}{\ell'}{j} > d'_j = B - \|\iter{\vx}{\ell'}\|_1$ (\inFull{like in the proof of Corollary~\ref{cor:plain_greedy}, }the last equality holds since the inequalities $\OPT_j - \iters{x}{\ell'}{j} > d'_j$ and $u_j \geq \OPT_j$ exclude the possibility of $d'_j = u_j - \iters{x}{\ell'}{j}$). The penultimate inequality holds since, by our assumption, $B - \OPT_j \geq \|\vy\|_1 \geq B / 2$.

It remains to consider the caes in which $\|\vy\|_1 \leq B / 2$. In this case, for every coordinate $i \in [n]$ we have $y_i - x_i \leq y_i \leq B/2 \leq B - \|\vx\|_1$ as long as $\|\vx\|_1 \leq B/2$. Thus, all the iterations of Algorithm~\ref{alg:greedy} are good until $\|\vx\|_1$ gets to a size lager than $B / 2$; which implies $\|\iter{\vx}{\ell}\|_1 \geq B / 2$. Hence, Lemma~\ref{lem:conditioned_guarantee} and Observation~\ref{obs:increase} allow us to lower bound $F(\outputCa)$ also by
\begin{align*}
	F(\outputCa)
	\geq{} &
	F(\iter{\vx}{\ell})
	\geq
	(1 - e^{-\|\iter{\vx}{\ell}\|_1/(\|\vy\|_1 + \eps B)}) \cdot F(\OPT - \OPT_j \basic_j) - \|\iter{\vx}{\ell}\|_1 \cdot \eps L\\
	\geq{}&
	(1 - e^{-(B/2)/(B/2 + \eps B)}) \cdot F(\OPT - \OPT_j \basic_j) - \eps BL\\
	\geq{} &
	(1 - e^{2\eps - 1}) \cdot F(\OPT - \OPT_j \basic_j) - \eps BL
	\geq
	(1 - e^{- 1} - 2\eps) \cdot F(\OPT - \OPT_j \basic_j) - \eps BL
	\enspace,
\end{align*}
where the \inFull{third}\inShortened{second} inequality follows from the above discussion and the inequality $\|\iter{\vx}{\ell}| \leq B$ which holds since $\iter{\vx}{\ell}$ is a feasible solution.
\end{proof}

We are now ready to present the enhanced algorithm \AlgECA, which appears as Algorithm~\ref{alg:enhanced_greedy}. The enhancement done in this algorithm, and its analysis, is related to an algorithm of~\citet{cohen2008generalized} obtaining the same approximation guarantee for the special case of discrete monotone submodular functions.

\begin{algorithm}
\caption{\AlgECA$(\eps)$} \label{alg:enhanced_greedy}
Let $\outputCa$ be the solution produced by Algorithm~\ref{alg:greedy} when run with $\eps$.\\
Let $\outputCaP$ be the best solution among the $n + 1$ solutions $\outputCa$ and $\{u_i \cdot \basic_i\}_{i \in [n]}$. \\
\Return{$\outputCaP$}.
\end{algorithm}

It is clear that the time complexity of Algorithm~\ref{alg:enhanced_greedy} is dominated by the time complexity of Algorithm~\ref{alg:greedy}. Thus, we only need to analyze the approximation ratio of Algorithm~\ref{alg:enhanced_greedy}. This is done by the next theorem, whose proofs relies on the fact that one of the solutions checked by Algorithm~\ref{alg:enhanced_greedy} is $u_je_j$ for the coordinate $j$ whose existence is guaranteed by Lemma~\ref{lem:element_out}.

\begin{restatable}{theorem}{thmEnhancedGreedy} \label{thm:enhanced_greedy}
Algorithm~\ref{alg:enhanced_greedy} outputs a solution of value at least $\big(\frac{e - 1}{2e - 1} - 2\eps\big) \cdot F(OPT) - \eps BL \geq (0.387 - 2\eps) \cdot F(OPT) - \eps BL$. It has $O(n/\eps)$ iterations, each running in $O(n\sqrt{B/\eps} + n \log n)$ time, which yields a time complexity of $O(n^2\sqrt{B}/\eps^{1.5} + n^2 \log n / \eps)$.
\end{restatable}
\newcommand{\proofEnhancedGreedy}{%
\begin{proof}
By Lemma~\ref{lem:element_out}, there exists a coordinate $j \in [n]$ such that
\[
	F(g) \geq (1 - e^{-1} - 2\eps) \cdot F(\OPT - \OPT_j \basic_j) - \eps B L
	\enspace.
\]
Since Algorithm~\ref{alg:enhanced_greedy} picks a solution $\outputCaP$ that is at least as good as both $\outputCa$ and $u_j \basic_j$, the last inequality and the monotonicity of $F$ imply together that
{\allowdisplaybreaks
\begin{align*}
	F(\outputCaP)
	\geq{} &
	\frac{1}{2 - e^{-1} - 2\eps} \cdot F(g) + \frac{1 - e^{-1} - 2\eps}{2 - e^{-1} - 2\eps} \cdot F(u_j \basic_j)\\
	\geq{} &
	\frac{1}{2 - e^{-1} - 2\eps} \cdot \left[(1 - e^{-1} - 2\eps) \cdot F(\OPT - \OPT_j \basic_j) - \eps BL\right] + \frac{1 - e^{-1} - 2\eps}{2 - e^{-1} - 2\eps} \cdot F(\OPT_j \basic_j)\\
	\geq{} &
	\frac{1 - e^{-1} - 2\eps}{2 - e^{-1} - 2\eps} \cdot \left[F(\OPT - \OPT_j \basic_j) + F(\OPT_j \basic_j)\right] - \eps BL\\
	\geq{} &
	\frac{1 - e^{-1} - 2\eps}{2 - e^{-1} - 2\eps} \cdot F(\OPT) - \eps BL
	\geq
	\left(\frac{1 - e^{-1}}{2 - e^{-1}} - 2\eps\right) \cdot F(\OPT) - \eps BL\\
	={} &
	\left(\frac{e - 1}{2e - 1}  - 2\eps\right) \cdot F(\OPT) - \eps BL
	\enspace,
\end{align*}
}%
where the penultimate inequality holds by the submodularity of $F$.
\end{proof}%
}
\inFull{\proofEnhancedGreedy}
\section{Optimal Approximation Ratio}

In this section we describe a more involved way to enhance the plain coordinate ascent algorithm from Section~\ref{sec:plain}, which leads to the algorithm that we name \AlgFECA and achieves the optimal approximation ratio of 1 - 1/e (up to some error term). This enhancement uses as a black box an algorithm for a one dimensional optimization problem whose properties are summarized by the following proposition. We include the proof of this proposition in Appendix~\ref{app:one_coordinate_get_value}.

\begin{restatable}{proposition}{propOneCoordinateGetValue} \label{prop:one_coordinate_get_value}
Given a point $\vx \in [\vzero, \vu]$, a coordinate $i \in [n]$, a target value $F(\vx) \leq v \leq F(\vx \vee u_i \basic_i)$ and a positive parameter $\eps \in (0, 1)$, there is a polynomial time algorithm that runs in $O(\log(B / \eps))$ time and returns a value $0 \leq y \leq u_i - x_i$ such that
\begin{itemize}
	\item $F(\vx + y \basic_i) \geq v - \eps L$.
	\item There is no value $0 \leq y' < y$ such that $F(\vx + y' \basic_i) \geq v$.
\end{itemize}
\end{restatable}


We can now give a simplified version of \AlgFECA, which appears as Algorithm~\ref{alg:three_coordinate_enhancement_guess}. For simplicity, we assume in the description and analysis of this algorithm that $n \geq 3$. If this is not the case, one can simulate it by adding dummy coordinates that do not affect the value of the objective function. Algorithm~\ref{alg:three_coordinate_enhancement_guess} starts by guessing two coordinates $h_1$ and $h_2$ that contribute a lot of value to $\OPT$. Then it constructs a solution $\vx$ with a small support using two executions of the algorithm whose existence is guaranteed by Proposition~\ref{prop:one_coordinate_get_value}, one execution for each one of the coordinates $h_1$ and $h_2$. It then completes the solution $\vx$ into a full solution by executing Algorithm~\ref{alg:greedy} after ``contracting'' the coordinates $h_1$ and $h_2$, i.e., modifying the objective function so that it implicitly assumes that these coordinates take the values they take in $\vx$.
\begin{algorithm}
\caption{\AlgFECA \textsc{ (Simplified)} $(\eps)$} \label{alg:three_coordinate_enhancement_guess}
Guess the coordinate $h_1 \in [n]$ maximizing $F(\OPT_{h_1} \cdot \basic_{h_1})$ and the coordinate $h_2 \in [n] \setminus \{h_1\}$ other than $h_1$ maximizing $F(\sum_{i \in \{h_1, h_2\}} \OPT_i \cdot \basic_i)$.\\
Let $\vx \gets \vzero$.\\
\For{$i = 1$ \KwTo $2$}
{
	Guess a value $v_i$ obeying \[\max\{F(\vx), F(\vx + \OPT_{h_i}\basic_{h_i}) - \eps \cdot F(\OPT)\} \leq v_i \leq F(\vx + \OPT_{h_i}\basic_{h_i})\enspace.\]\\
	Let $y_i$ be the value returned by the algorithm whose existence is guaranteed by Proposition~\ref{prop:one_coordinate_get_value} given $\vx$ as the input vector, the coordinate $h_i$ and the target value $v_i$.\\
	Update $\vx \gets \vx + y_i \basic_{h_i}$.
}
Execute Algorithm~\ref{alg:greedy} on the instance obtained by removing the coordinates $h_1$ and $h_2$, replacing the objective function with $F'(\vx') = F(\vx' + \vx) - F(\vx)$ and decreasing $B$ by $\|\vx\|_1$. Let $\outputCa$ be the output of Algorithm~\ref{alg:greedy}.\\
Return $\vx + \outputCa$ (we denote this sum by $\outputCaPP$ in the analysis).
\end{algorithm}

We begin the analysis of Algorithm~\ref{alg:three_coordinate_enhancement_guess} by bounding its time complexity.
\begin{restatable}{observation}{obsGuessTimeComplexity} \label{obs:guess_time_complexity}
Assuming the guesses made by Algorithm~\ref{alg:three_coordinate_enhancement_guess} do not require any time, Algorithm~\ref{alg:three_coordinate_enhancement_guess} has $O(n/\eps)$ iterations, each running in $O(n\sqrt{B/\eps} + n \log n)$ time, yielding a time complexity of $O(n^2\sqrt{B}/\eps^{1.5} + n^2 \log n / \eps)$.
\end{restatable}
\newcommand{\proofGuessTimeComplexity}{%
\begin{proof}
Besides the two executions of the algorithm whose existence is guaranteed by Proposition~\ref{prop:one_coordinate_get_value} and the execution of Algorithm~\ref{alg:greedy}, Algorithm~\ref{alg:three_coordinate_enhancement_guess} uses only constant time. Thus, the time complexity of Algorithm~\ref{alg:three_coordinate_enhancement_guess} is upper bounded by the sum of the time complexities of the two other algorithms mentioned. Furthermore, by Proposition~\ref{prop:one_coordinate_get_value}, the total time complexity of the algorithm whose existence is guaranteed by this proposition is only
\[
	O\left(\log \left(\frac{B}{\eps}\right)\right)
	\enspace,
\]
which is upper bounded by the time complexity of a single iteration of Algorithm~\ref{alg:greedy} as given by Observation~\ref{obs:complexity_greedy}. Hence, both the number of iterations and the time per iteration of Algorithm~\ref{alg:three_coordinate_enhancement_guess} are asymptotically identical to the corresponding values for Algorithm~\ref{alg:greedy}.
\end{proof}%
}
\inFull{\proofGuessTimeComplexity}

The next step in the analysis of Algorithm~\ref{alg:three_coordinate_enhancement_guess} is proving some properties of the vector $\vx = \sum_{i \in \{h_1, h_2\}} y_i \cdot \basic_i$ produced by the first part of the algorithm. In a nutshell, these properties holds since the definition of $v_j$ and the properties of Proposition~\ref{prop:one_coordinate_get_value} show together that the value chosen for $x_{h_j}$ by the algorithm of Proposition~\ref{prop:one_coordinate_get_value} gives almost as much value as choosing $\OPT_{h_j}$, but it never overestimates $\OPT_{h_j}$.
\begin{restatable}{lemma}{lemGuessingProperties} \label{lem:guessing_properties}
$F(\vx) \geq F(\sum_{j = 1}^2 \OPT_{h_j} \cdot \basic_{h_j}) - 2\eps \cdot F(OPT) - 2\eps L$ and $\vx \leq \sum_{j = 1}^2 \OPT_{h_j} \cdot \basic_{h_j}$.
\end{restatable}
\newcommand{\proofGuessingProperties}{%
\begin{proof}
Recall that the support of $\vx$ contains only the coordinates $h_1$ and $h_2$. Thus, to prove the lemma, it suffices to argue that for every $i \in \{1, 2\}$
\begin{align} \label{eq:value_gain}
	F\left(\sum_{j = 1}^i x_{h_j} \cdot \basic_{h_j}\right) - F&\left(\sum_{j = 1}^{i - 1} x_{h_j} \cdot \basic_{h_j}\right)\\\nonumber
	\geq{} &
	F\left(\sum_{j = 1}^i \OPT_{h_j} \cdot \basic_{h_j}\right) - F\left(\sum_{j = 1}^{i - 1} \OPT_{h_j} \cdot \basic_{h_j}\right) - \eps \cdot F(\OPT) - \eps L
\end{align}
and
\begin{equation} \label{eq:lower_bounded_coordinate-wise}
	x_i \leq \OPT_i
	\enspace.
\end{equation}
We prove this by induction on $i$. In other words, we prove that the two above inequalities hold for $i \in \{1, 2\}$ given that they holds for every $i' < i$ that belongs to $\{1, 2\}$ (if there is such an $i'$).

The value of $x_i$ is determined by an execution of the algorithm whose existence is guaranteed by Proposition~\ref{prop:one_coordinate_get_value}. Thus, to prove the above inequalities, we need to use the guarantees of this proposition. Moreover, we notice that this is possible since the target value $v_i$ passed to the algorithm of this proposition clearly falls within the allowed range because $\OPT_{h_i} \leq u_{h_i}$. Hence, by the first guarantee of Proposition~\ref{prop:one_coordinate_get_value},
\[
	F\left(\sum_{j = 1}^i x_{h_j} \cdot \basic_j\right)
	\geq
	v_i - \eps L
	\geq
	F\left(\sum_{j = 1}^{i - 1} x_{h_j} \cdot \basic_{h_j} + \OPT_{h_i}\basic_{h_i}\right) - \eps \cdot F(\OPT) - \eps L
	\enspace.
\]
Inequality~\eqref{eq:value_gain} now follows from the last inequality by subtracting $F\left(\sum_{j = 1}^{i - 1} x_{h_j} \cdot \basic_{h_j}\right)$ from both its sides and observing that, by the submodularity of $F$ and the induction hypothesis,
\[
	F\mspace{-1mu}\left(\sum_{j = 1}^{i - 1} x_{h_j} \cdot \basic_j + \OPT_{h_i}\basic_{h_i}\mspace{-2mu}\right) - F\mspace{-1mu}\left(\sum_{j = 1}^{i - 1} x_{h_j} \cdot \basic_j\mspace{-2mu}\right)
	\mspace{-2mu}\geq\mspace{-1mu}
	F\mspace{-1mu}\left(\sum_{j = 1}^{i} \OPT_{h_j} \cdot \basic_{h_j}\mspace{-2mu}\right) - F\mspace{-1mu}\left(\sum_{j = 1}^{i - 1} \OPT_{h_j} \cdot \basic_{h_j}\mspace{-2mu}\right)
	.
\]

To prove Inequality~\eqref{eq:lower_bounded_coordinate-wise}, we note that the second guarantee of Proposition~\ref{prop:one_coordinate_get_value} implies that for every $0 \leq y < x_{h_j}$ we have
\[
	F\left(\sum_{j = 1}^{i - 1} x_{h_j} \cdot \basic_j + y\basic_{h_i}\right)
	<
	v_i
	\leq
	F\left(\sum_{j = 1}^{i - 1} x_{h_j} \cdot \basic_j + \OPT_j\basic_{h_i}\right)
	\enspace,
\]
and therefore, $\OPT_{h_j}$ cannot fall in the range $[0, x_{h_j})$.
\end{proof}%
}
\inFull{\proofGuessingProperties}

We now ready to prove the approximation guarantee of Algorithm~\ref{alg:three_coordinate_enhancement_guess}. Intuitively, this proof is based on simply adding up the lower bound on $F(\vx)$ given by Lemma~\ref{lem:guessing_properties} and the lower bound on $F'(\outputCa)$ given by Lemma~\ref{lem:element_out}. Some of the ideas used in the proof can be traced back to a recent result by~\citet{perosnal2020nutov}, who described an algorithm achieving $(1 - 1/e)$-approximation for the discrete version of the problem we consider (namely, maximizing a non-negative monontone discrete submodular function subject to a knapsack constraint) using $O(n^4)$ function evaluations. 

\begin{restatable}{lemma}{lemApproximationRatio} \label{lem:approximation_ratio}
Algorithm~\ref{alg:three_coordinate_enhancement_guess} outputs a vector $\outputCaPP$ whose value is at least $(1 - 1/e - 4\eps) \cdot F(\OPT) - \eps (B + 2)L$.
\end{restatable}
\newcommand{\proofApproximationRatio}{%
\begin{proof}
Since $\vx \leq \sum_{j = 1}^2 \OPT_{h_j} \cdot \basic_{h_j}$ by Lemma~\ref{lem:guessing_properties}, the submodularity of $F$ guarantees that
\begin{align*}
	F'\left(\OPT - \sum_{j = 1}^2 \OPT_{h_j} \cdot \basic_{h_j}\right)
	={} &
	F\left(\OPT + \sum_{j = 1}^2 (x_{h_j} - \OPT_{h_j}) \cdot \basic_{h_j}\right) - F(\vx)\\
	\geq{} &
	F(\OPT) - F\left(\sum_{j = 1}^2 \OPT_{h_j} \cdot \basic_{h_j}\right)
	\enspace.
\end{align*}

Therefore, since $\OPT - \sum_{j = 1}^2 \OPT_{h_j} \cdot \basic_{h_j}$ is one feasible solution for the instance received by Algorithm~\ref{lem:element_out}, we get by Lemma~\ref{lem:element_out} that there exists a coordinate $i \in [n] \setminus \{h_1, h_2\}$ such that\footnote{As stated, Lemma~\ref{lem:element_out} applies only to the optimal solution, not to every feasible solution. However, one can verify that its proof does not use the optimality of the solution.}
{\allowdisplaybreaks
\begin{align*}
	F'(\outputCa)
	\geq{} &
	(1 - 1/e - 2\eps) \cdot F'\left(\OPT - \sum_{j = 1}^2 \OPT_{h_j} \cdot \basic_{h_j} - \OPT_i \basic_i\right) - \eps BL\\
	\geq{} &
	(1 - 1/e - 2\eps) \cdot \left[F'\left(\OPT - \sum_{j = 1}^2 \OPT_{h_j} \cdot \basic_{h_j}\right) - F'(\OPT_i \basic_i)\right] - \eps BL\\
	\geq{} &
	(1 - 1/e - 2\eps) \cdot \left[F(\OPT) - F\left(\sum_{j = 1}^2 \OPT_{h_j} \cdot \basic_{h_j}\right) - F'(\OPT_i \basic_i)\right] - \eps BL\\
	\geq{} &
	(1 - 1/e - 2\eps) \cdot \left[F(\OPT) - \frac{3}{2} \cdot F\left(\sum_{j = 1}^2 \OPT_{h_j} \cdot \basic_{h_j}\right)\right] - \eps BL
	\enspace,
\end{align*}
where the second inequality follows from the submodularity of $F$, and the last inequality holds since the submodularity of $F$ and the definitions of $h_1$ and $h_2$ imply
\begin{align*}
	F'(\OPT_i \basic_i)
	={} &
	F\left(\sum_{j = 1}^2 \OPT_{h_j} \cdot \basic_{h_j} + \OPT_i \basic_i\right) - F\left(\sum_{j = 1}^2 \OPT_{h_j} \cdot \basic_{h_j}\right)\\
	\leq{} &
	\frac{1}{2}\left[F(\OPT_i \basic_i) - F(\vzero) + F(\OPT_{h_1} + \OPT_i \basic_i) - F(\OPT_{h_1}\basic_{h_1})\right]\\
	\leq{} &
	\frac{1}{2}\left[F(\OPT_{h_1} \basic_{h_1}) - F(\vzero) + F(\OPT_{h_1} + \OPT_{h_2} \basic_{h_2}) - F(\OPT_{h_1}\basic_{h_1})\right]\\
	={} &
	\frac{1}{2}\left[F(\OPT_{h_1}\basic_{h_1} + \OPT_{h_2} \basic_{h_2}) - F(\vzero)\right]
	\leq
	\frac{F(\OPT_{h_1}\basic_{h_1} + \OPT_{h_2} \basic_{h_2})}{2}
	\enspace.
\end{align*}
}%

We are now ready to calculate the value of $\outputCaPP = \vx + \outputCa$. By the above calculation and Lemma~\ref{lem:guessing_properties},
\begin{align*}
	F(\vx + \outputCa)
	={} &
	F'(\outputCa) + F(\vx)\\
	\geq{} &
	(1 - 1/e - 2\eps) \cdot \left[F(\OPT) - \frac{3}{2} \cdot F\left(\sum_{j = 1}^2 \OPT_{h_j} \cdot \basic_{h_j}\right)\right] - \eps BL \\&+ F\left(\sum_{j = 1}^2 \OPT_{h_j} \cdot \basic_{h_j}\right) - 2\eps \cdot F(OPT) - 2\eps L\\
	\geq{} &
	(1 - 1/e - 4\eps) \cdot F(\OPT) - \eps (B + 2)L
	\enspace.
	\tag*{\qedhere}
\end{align*}
\end{proof}%
}
\inFull{\proofApproximationRatio}

To get our final \AlgFECA algorithm, we need to explain how to implement the guesses of Algorithm~\ref{alg:three_coordinate_enhancement_guess}. The coordinates $h_1$ and $h_2$ can be guessed by simply iterating over all the possible pairs of two coordinates. Similarly, by the next observation, to get $v_i$ it suffices to try all the possible values in the set $\{F(\vx) + \eps j \cdot F(u_{h_i}\basic_{h_i}) \mid \text{$j$ is a non-negative integer and } F(\vx) + \eps j \cdot F(u_{h_i}\basic_{h_i}) \leq F(\vx + u_{h_i}\basic_{h_i})\}$. In the following, we refer to this set as $\cJ(\vx, h_i)$.

\begin{restatable}{observation}{obsGuessSet} \label{obs:guess_set}
Consider the vector $\vx$ at the point in which Algorithm~\ref{alg:three_coordinate_enhancement_guess} guesses the value $v_i$. Then, there exists a value in the set $\cJ(\vx, h_i)$ obeying the requirements from $v_i$.
\end{restatable}
\newcommand{\proofGuessSet}{%
\begin{proof}
Let $j$ be the maximal integer for which $F(\vx) + \eps j \cdot F(u_{h_i}\basic_{h_i}) \leq F(\vx + \OPT_{h_i} \basic_{h_i}) \leq F(\vx + u_{h_i} \basic_{h_i})$. Since $F(\vx) \leq F(\vx + \OPT_{h_i} \basic_{h_i})$ by the monotonicity of $F$, $j$ is non-negative, and thus, $F(\vx) + \eps j \cdot F(u_{h_i}\basic_{h_i})$ belongs to $\cJ(\vx, h_i)$ and $F(\vx) + \eps j \cdot F(u_{h_i}\basic_{h_i}) \geq F(\vx)$. Furthermore, by the definition of $j$,
\[
	F(\vx) + \eps j \cdot F(u_{h_i}\basic_{h_i})
	\geq
	F(\vx + \OPT_{h_i} \basic_{h_i}) - \eps \cdot F(u_{h_i}\basic_{h_i})
	\geq
	F(\vx + \OPT_{h_i} \basic_{h_i}) - \eps \cdot F(\OPT)
	\enspace,
\]
where the second inequality holds since $u_{h_i}\basic_{h_i}$ is a feasible solution. Thus, $F(\vx) + \eps j \cdot F(u_{h_i}\basic_{h_i})$ obeys the requirements from $v_i$.
\end{proof}%
}
\inFull{\proofGuessSet}

Our final \AlgFECA algorithm appears as Algorithm~\ref{alg:three_coordinate_enhancement}. By the above discussion, the number of iterations it makes exceeds the number of iterations given by Observation~\ref{obs:guess_time_complexity} only by a factor of
\begin{align*}
	n^2 \cdot \prod_{i = 1}^2 |\cJ(\vx, h_i)|
	\leq{} &
	n^2 \cdot\prod_{i = 1}^2\left(1 + \frac{F(\vx + u_{h_i}\basic_{h_i}) - F(\vx)}{\eps \cdot F(u_{h_i}\basic_{h_i})}\right)
	=
	O(\eps^{-2}n^2)
	\enspace,
\end{align*}
where the equality holds since the submodulrity and non-negativity of $f$ imply $F(\vx + u_{h_i}\basic_{h_i}) - F(\vx) \leq F(u_{h_i}\basic_{h_i})$.

\begin{algorithm}
\caption{\AlgFECA $(\eps)$} \label{alg:three_coordinate_enhancement}
\For{every pair of distinct coordinates $h_1, h_2 \in [n]$}
{
	Let $\vx^{(0)} \gets \vzero$.\\
	\For{every $v_1 \in \cJ(\vx^{(0)}, h_1)$} 
	{
		Let $y_1$ be the value returned by the algorithm guaranteed by Proposition~\ref{prop:one_coordinate_get_value} given $\vx^{(0)}$ as the input vector, the coordinate $h_1$ and the target value $v_1$.\\
		Set $\vx^{(1)} \gets \vx^{(0)} + y_1 \basic_{h_1}$.\\
		\For{every $v_2 \in \cJ(\vx^{(1)}, h_2)$}
		{
			Let $y_2$ be the value returned by the algorithm guaranteed by Proposition~\ref{prop:one_coordinate_get_value} given $\vx^{(1)}$ as the input vector, the coordinate $h_2$ and the target value $v_2$.\\
			Update $\vx \gets \vx^{(1)} + y_2 \basic_{h_2}$.\\
			Execute Algorithm~\ref{alg:greedy} on the instance obtained by removing the coordinates $h_1$ and $h_2$, replacing the objective function with $F'(\vx') = F(\vx' + \vx) - F(\vx)$ and decreasing $B$ by $\|\vx\|_1$. Let $\outputCa$ be the output of Algorithm~\ref{alg:greedy}.\\
			Mark $\vx + \outputCa$ as a candidate solution.
		}
	}
}
Return the solution maximizing $F$ among all the solutions marked above as candidate solutions.
\end{algorithm}

The next theorem summarizes the result we have proved in this section.
\begin{theorem}
For every $\eps \in (0, 1)$, Algorithm~\ref{alg:three_coordinate_enhancement} is an algorithm for our problem which produces a solution of value at least $(1 - 1/e - 4\eps) \cdot F(\OPT) - \eps (B + 2)L$. It has $O(n^3/\eps^3)$ iterations, each running in $O(n\sqrt{B/\eps} + n \log n)$ time, which yields a time complexity of $O(n^4\sqrt{B}/\eps^{2.5} + n^4 \log n / \eps^3)$.
\end{theorem}

In all the loops of Algorithm~\ref{alg:three_coordinate_enhancement}, the iterations are independent, and thus, can be done in parallel instead of sequentially. Thus, the parallel time required for Algorithm~\ref{alg:three_coordinate_enhancement} is equal to the time complexity of Algorithm~\ref{alg:three_coordinate_enhancement_guess}, which by Observation~\ref{obs:guess_time_complexity} is only $O(n^2\sqrt{B}/\eps^{1.5} + n^2 \log n / \eps)$.

\section{Conclusion}
In this paper, we provided the first constant factor approximation guarantees for the problem of maximizing a monotone continuous submodular function subject to a linear constraint. Crucially, our results did not rely on DR-submodularity. 

\bibliographystyle{plainnat}
\bibliography{References}

\appendix

\section{Proof of Proposition~\ref*{prop:one_coordinate_ratio}} \label{app:one_coordinate_ratio}

In this section we prove Proposition~\ref{prop:one_coordinate_ratio}. Let us begin by restating this proposition.

\propOneCoordinateRatio*

As might be expected, the algorithm we use to prove Proposition~\ref{prop:one_coordinate_ratio} tries a relatively small set of possible options for $y$, and then outputs the value yielding the maximum $F(\vx + y\basic_i) / y$ ratio. To define the set of values which the algorithm checks, we first need to define the following recursive series.
\[
	z_0 = a
	\qquad
	\text{and}
	\qquad
	z_i = z_{i - 1} + \sqrt{\eps z_{i - 1}}
	\enspace.
\]

Using this definition, we can now formally state the algorithm used to prove Proposition~\ref{prop:one_coordinate_ratio} as Algorithm~\ref{alg:one_coordinate_ratio}.

\begin{algorithm}
\caption{\textsc{One Coordinate Ratio Maximizer}} \label{alg:one_coordinate_ratio}
Let $M = \{b\} \cup \{z_i \mid \text{$i$ is a non-negative integer and } z_i \in [a, b]\}$.\\
Return $y \in \arg \max_{y \in M} F(\vx + y\basic_i) / y$.
\end{algorithm}

Before analyzing the quality of the solution returned by Algorithm~\ref{alg:one_coordinate_ratio}, let us prove that it indeed has the required time complexity.
\begin{lemma}
The time complexity of Algorithm~\ref{alg:one_coordinate_ratio} is $O(\sqrt{B/\eps} + \log (\eps / a))$.
\end{lemma}
\begin{proof}
Let $\ell$ be the smallest non-negative integer such that $z_\ell \geq B$. Clearly, the size of $M$ is upper bounded by $\ell + 3$ since $b \leq u_i \leq B$, and thus, the time complexity of Algorithm~\ref{alg:one_coordinate_ratio} is $O(\ell)$. Hence, to prove the lemma, it suffices to argue that there exists a positive integer $i' = O(\sqrt{B / \eps} + \log (\eps / a))$ such that $z_{i'} \geq B$, and thus, $O(\ell) = O(i') = O(\sqrt{B / \eps} + \log (\eps / a))$.

Observe that if $z_i \leq \eps$, then $z_{i + 1} \geq 2z_i$. Thus, for $i_0 = \lceil \log_2 (\eps / a)\rceil$ we already get $z_{i_0} \geq \eps$. Consider now the function $f(x) = \eps (x^2 + 16) / 16$. We would like to prove by induction that for every non-negative integer $i \geq 0$ we have $f(i) \leq z_{i + i_0}$. For $i = 0$ this holds since $f(0) = \eps \leq z_{i_0}$. Assume now that this claim holds for some integer $i - 1 \geq 0$, and let us prove it for $i$. Since $z_{i - 1 + i_0} \geq f(i - 1)$ by the induction hypothesis, it suffices to argue that $f(i) - f(i - 1) \leq z_{i + i_0} - z_{i - 1 + i_0} = \sqrt{\eps z_{i - 1 + i_0}}$. By the definition of $f$,
\begin{align*}
	f(i) - f(&i - 1)
	=
	\frac{\eps (i^2 + 16)}{16} - \frac{\eps [(i - 1)^2 + 16]}{16}
	=
	\frac{\eps(2i - 1)}{16}\\
	\leq{} &
	\frac{\eps \cdot 4\sqrt{(i - 1)^2 + 16}}{16}
	=
	\sqrt{\eps \cdot \frac{\eps[(i - 1)^2 + 16]}{16}}
	=
	\sqrt{\eps \cdot f(i - 1)}
	\leq
	\sqrt{\eps \cdot z_{i - 1 + i_0}}
	\enspace,
\end{align*}
where the second inequality follows from the induction hypothesis, and the first inequality holds since for $i \geq 1$
\begin{align*}
	2i - 1 \leq 4\sqrt{(i - 1)^2 + 16}
	\iff{} &
	(2i - 1)^2 \leq 16[(i - 1)^2 + 16]\\
	\iff{} &
	4i^2 - 4i + 1 \leq 16i^2 - 32i + 272
	\iff
	0 \leq 12i^2 - 28i + 271
	\enspace,
\end{align*}
and the last inequality holds for every $i$.

To complete the proof, it remains to observe that for $i' = \lceil \log_2 (\eps / a)\rceil + \lceil 4\sqrt{B / \eps} \rceil = i_0 + \lceil 4\sqrt{B / \eps} \rceil$ we have
\[
	z_{i'}
	\geq
	f(i' - i_0)
	=
	\frac{\eps((i' - i_0)^2 + 16)}{16}
	\geq
	\frac{\eps(16B / \eps + 16)}{16}
	\geq
	B
	\enspace.
	\tag*{\qedhere}
\]
\end{proof}

Our next objective is to show that the solution produced by Algorithm~\ref{alg:one_coordinate_ratio} approximately maximizes the ratio $F(\vx + y\basic_i) / y$ within the range $[a, b]$. Let $y^*$ be a value within this range that truly maximizes this ratio, and let $y_M$ be the largest value in the set $M$ which is not larger than $y^*$ (possibly $y^* = y_M$ if $y^* \in M$). We argue below that $F(\vx + y^*\basic_i) / y^* \leq F(\vx + y_M\basic_i) / y_M + \eps L$, which completes the proof of Proposition~\ref{prop:one_coordinate_ratio} since the membership of $y_M$ in $M$ implies that the ratio $F(\vx + y\basic_i) / y$ for the value $y$ returned by Algorithm~\ref{alg:one_coordinate_ratio} is at least as good as $F(\vx + y_M\basic_i) / y_M$.

The next lemma gives us a simple upper bound on the ratio $F(\vx + y^*\basic_i) / y^*$.
\begin{lemma}
\[
	\frac{F(\vx + y^*\basic_i)}{y^*}
	\leq
	\frac{F(\vx + y_M\basic_i)}{y_M} + \frac{(y^* - y_M)^2L}{2y_M}
	\enspace.
\]
\end{lemma}
\begin{proof}
The derivative of $F(\vx + y\basic_i) / y$ by $y$ is
\[
	\frac{\frac{dF}{dy}(\vx + y\basic_i) \cdot y - F(\vx + y\basic_i)}{y^2}
	\enspace.
\]
Since $y^*$ is a maximizer of this ratio, the above derivative must be zero in $y^*$, i.e., we get
\[
	\frac{dF}{dy}(\vx + y^*\basic_i) \cdot y^* - F(\vx + y^*\basic_i) = 0
	\Rightarrow
	\frac{dF}{dy}(\vx + y^*\basic_i) = \frac{F(\vx + y^*\basic_i)}{y^*}
	\enspace.
\]
This allows us to use the smoothness of $F$ to upper bound the difference between $F(\vx + y^*\basic_i)$ and $F(\vx + y_M\basic_i)$ by
\begin{align*}
	F(\vx + y^*\basic_i) - F(\vx + y_M\basic_i&)
	=
	\int_{y_M}^{y^*} \frac{dF}{dy}(\vx + y\basic_i) dy
	\leq
	\int_{y_M}^{y^*} \left[\frac{dF}{dy}(\vx + y^*\basic_i) + (y^* - y)L\right] dy\\
	={} &
	\int_{y_M}^{y^*} \left[\frac{F(\vx + y^*\basic_i)}{y^*} + (y^* - y)L\right] dy
	=
	\frac{y^* - y_M}{y^*} \cdot F(\vx + y^*\basic_i) + \frac{(y^* - y_M)^2L}{2}
	\enspace.
\end{align*}
Rearranging the last inequality, we get
\[
	\frac{y_M}{y^*} \cdot F(\vx + y^*\basic_i)
	\leq
	F(\vx + y_M\basic_i) + \frac{(y^* - y_M)^2L}{2}
	\enspace,
\]
and the observation follows by dividing the last inequality by $y_M$.
\end{proof}

Given the above discussion, the last lemma implies that to prove Proposition~\ref{prop:one_coordinate_ratio} we only need to argue that $\frac{(y^* - y_M)^2L}{2y_M}$ is always upper bounded by $\eps L$. The following observation shows that this is indeed the case.

\begin{observation}
$\frac{(y^* - y_M)^2}{2y_M} \leq \eps$.
\end{observation}
\begin{proof}
By the definition of the set $M$, the value of $y^*$ must be at most $y_M + \sqrt{\eps y_M}$. Thus,
\[
	\frac{(y^* - y_M)^2}{2y_M}
	\leq
	\frac{(\sqrt{\eps y_M})^2}{2y_M}
	=
	\frac{\eps y_M}{2y_M}
	=
	\frac{\eps}{2}
	<
	\eps
	\enspace.
	\tag*{\qedhere}
\]
\end{proof}
\section{Proof of Proposition~\ref*{prop:one_coordinate_get_value}} \label{app:one_coordinate_get_value}

In this section we prove Proposition~\ref{prop:one_coordinate_get_value}. Let us begin by restating this proposition.

\propOneCoordinateGetValue*

The algorithm we use to prove Proposition~\ref{prop:one_coordinate_get_value} has two phases. In the first phase, the algorithm uses binary search to zoom in on a small range of $y$ values which includes the lowest $y$ value for which $F(\vx + y \basic_i) = v$. Then, in the second phase, the algorithm uses linear interpolation to pick a value $y$ from this range for which $F(\vx + y \basic_i)$ is close to $v$. The linear interpolation parameters have to be selected with care to make sure that the value picked obeys the second guarantee of the proposition. A formal statement of the algorithm appears as Algorithm~\ref{alg:one_coordinate_get_value}.

\begin{algorithm}
\DontPrintSemicolon
\caption{\textsc{One Coordinate Getting Target Value}} \label{alg:one_coordinate_get_value}
Let $a = 0$ and $b = u_i - x_i$.\\
\While{$b - a \geq \eps$}
{
	Let $m = (b - a) / 2$.\\
	\lIf{$F(\vx + m \basic_i) \geq v$}{Update $b \gets m$.}
	\lElse{Update $a \gets m$.}
}
Let $d \gets \frac{F(\vx + b \basic_i) - F(\vx + a \basic_i)}{b - a} + \frac{\eps L}{2}$, and $r \gets \frac{v - F(\vx + a \basic_i)}{d}$.\\
Return $a + r$.
\end{algorithm}

We begin the analysis of Algorithm~\ref{alg:one_coordinate_get_value} by showing that has the time complexity guaranteed by Proposition~\ref{prop:one_coordinate_get_value}.

\begin{observation}
The time complexity of Algorithm~\ref{alg:one_coordinate_get_value} is at most $O(\log (B / \eps))$.
\end{observation}
\begin{proof}
The time complexity of Algorithm~\ref{alg:one_coordinate_get_value} is proportional to the number of iterations made by the binary search in the first phase of the algorithm. Since this binary search starts with a range of size $u_i - x_i \leq u_i \leq B$, and ends when its range shrinks to a size of $\eps$ or less, the number of iterations it performs is upper bounded by $\lceil \log (B / \eps) \rceil$.
\end{proof}

Let us denote now by $a_0$ and $b_0$ the values of the variables $a$ and $b$ when the binary search phase of Algorithm~\ref{alg:one_coordinate_get_value} terminates. By the design of the binary search, it is clear that $F(\vx + a_0 \basic_i) \leq v$. Furthermore, this inequality can hold as an equality only when $a_0 = 0$. Let us now get bounds on the derivative of $F(\vx + y\basic_i)$ as a function of $y$ within the range $[a_0, b_0]$.
\begin{lemma} \label{lem:derivative_bound}
For every $y' \in [a_0, b_0]$, $\frac{dF}{dy}(\vx + y \basic_i) \in [d - \eps L, d]$.
\end{lemma}
\begin{proof}
By the smoothness of the function $F$
{\allowdisplaybreaks
\begin{align*}
	F(\vx + b_0 \basic_i) - F(\vx + a_0 \basic_i)
	={} &
	\int_{a_0}^{b_0} \frac{dF}{dy}(\vx + y \basic_i) dy
	\leq
	\int_{a_0}^{b_0} \left[\frac{dF}{dy}(\vx + y' \basic_i) + |y - y'|L\right] dy\\
	\leq{} &
	(b_0 - a_0) \cdot \frac{dF}{dy}(\vx + y' \basic_i) + \frac{(b_0 - a_0)^2L}{2}\\
	\leq{} &
	(b_0 - a_0) \cdot \frac{dF}{dy}(\vx + y' \basic_i) + \frac{(b_0 - a_0)\eps L}{2}
	\enspace.
\end{align*}
}
Dividing the last inequality by $b_0 - a_0$, we get
\[
	d - \frac{\eps L}{2}
	\leq
	\frac{dF}{dy}(\vx + y' \basic_i) + \frac{\eps L}{2}
	\implies
	d - \eps L
	\leq
	\frac{dF}{dy}(\vx + y' \basic_i)
	\enspace.
\]

Similarly, the smoothness of $F$ also implies
\begin{align*}
	F(\vx + b_0 \basic_i) - F(\vx + a_0 \basic_i)
	={} &
	\int_{a_0}^{b_0} \frac{dF}{dy}(\vx + y \basic_i) dy
	\geq
	\int_{a_0}^{b_0} \left[\frac{dF}{dy}(\vx + y' \basic_i) - |y - y'|L\right] dy\\
	\geq{} &
	(b_0 - a_0) \cdot \frac{dF}{dy}(\vx + y' \basic_i) - \frac{(b_0 - a_0)^2L}{2}\\
	\geq{} &
	(b_0 - a_0) \cdot \frac{dF}{dy}(\vx + y' \basic_i) - \frac{(b_0 - a_0)\eps L}{2}
	\enspace,
\end{align*}
and this time dividing the last inequality by $b_0 - a_0$ yields
\[
	d - \frac{\eps L}{2}
	\geq
	\frac{dF}{dy}(\vx + y' \basic_i) - \frac{\eps L}{2}
	\implies
	d
	\geq
	\frac{dF}{dy}(\vx + y' \basic_i)
	\enspace.
	\tag*{\qedhere}
\]
\end{proof}

The following corollary now completes the proof of Proposition~\ref{prop:one_coordinate_get_value} since the output of Algorithm~\ref{alg:one_coordinate_get_value} is $a_0 + r$.

\begin{corollary}
$F(\vx + (a_0 + r)\basic_i) \geq v - \eps L$, and furthermore, $F(\vx + y'\basic_i) < v$ for every $0 \leq y' < a_0 + r$.
\end{corollary}
\begin{proof}
Using Lemma~\ref{lem:derivative_bound}, we get
\begin{align*}
	F(\vx + (a_0 + r)\basic_i)
	={} &
	F(\vx + a_0 \basic_i) + \int_{a_0}^{a_0 + r} \frac{dF}{dy}F(\vx + y\basic_i) dy\\
	\geq{} &
	F(\vx + a_0 \basic_i) + \int_{a_0}^{a_0 + r} (d - \eps L) dy
	=
	F(\vx + a_0 \basic_i) + r(d - \eps L)\\
	={} &
	F(\vx + a_0 \basic_i) + [v - F(\vx + a_0 \basic_i)] - r\eps L
	\geq
	F(\vx + a_0 \basic_i) + [v - F(\vx + a_0 \basic_i)] - \eps L
	\enspace,
\end{align*}
where the third equality holds by plugging in the definition of $r$, and the last inequality holds since
\[
	r
	=
	\frac{v - F(\vx + a \basic_i)}{d}
	\leq
	\frac{F(\vx + b \basic_i) - F(\vx + a \basic_i)}{d}
	\leq
	b - a
	\leq
	\eps
	<
	1
	\enspace.
\]

Similarly, Lemma~\ref{lem:derivative_bound} also implies for every $a_0 \leq y' < a_0 + r$
\begin{align*}
	F(\vx + y' \basic_i)
	={} &
	F(\vx + a_0 \basic_i) + \int_{a_0}^{y'} \frac{dF}{dy}F(\vx + y\basic_i) dy\\
	\leq{} &
	F(\vx + a_0 \basic_i) + \int_{a_0}^{y'} d\; dy
	=
	F(\vx + a_0 \basic_i) + (y' - a_0)d\\
	<{} &
	F(\vx + a_0 \basic_i) + rd
	=
	F(\vx + a_0 \basic_i) + [v - F(\vx + a_0 \basic_i)]
	=
	v
	\enspace.
\end{align*}
If $r > 0$, then the last inequality completes the proof of the corollary because the monotonicity of $F$ guarantees $F(\vx + y' \basic_i) \leq F(\vx + a_0 \basic_i) < v$ for every $0 \leq y' < a_0$. Thus, it remains to consider the case of $r = 0$. This case happens only when $F(\vx + a_0\basic_i) = v$, which implies by the discussion before Lemma~\ref{lem:derivative_bound} that $a_0 = 0$ as well. Hence, the requirement $F(\vx + y'\basic_i) < v$ for every $0 \leq y' < a_0 + r$ is trivial in this case.
\end{proof}
\inShortened{\input{MissingProofs}}

\end{document}